\documentclass[11pt]{article}
\usepackage{latexsym,color,amsmath,amsthm,amssymb,amscd,amsfonts}

\newtheorem{lemma}{Lemma}[section]

\newtheorem{corollary}[lemma]{Corollary}

\newtheorem{conj*}[lemma]{Conjecture}

\newtheorem*{remark*}{Remark}

\makeatletter \@addtoreset {equation}{section}

\renewcommand\theequation
  {\ifnum \c@subsection>\z@ \arabic{section}.\arabic{subsection}.\arabic{equation}
  \else \arabic{section}.\arabic{equation} \fi}
\makeatother
\begin{document}
\title{Localization of Multi-Dimensional Wigner Distributions}
\author{Elliott H. Lieb \ and \ Yaron Ostrover}
\date{August 10, 2010
\thanks{\copyright\, 2010 by the authors. This paper may be reproduced,
in its entirety, for non-commercial purposes.}}
\maketitle
\begin{abstract}
  A well known result of P. Flandrin states that a Gaussian uniquely
  maximizes the integral of the Wigner distribution over every centered disc
  in the phase plane.  While there is no difficulty in generalizing
  this result to higher-dimensional poly-discs, the generalization to
  balls is less obvious. In this note we provide such a generalization.
\end{abstract}

\section{Introduction}
The Wigner quasi-probability distribution was introduced by
Wigner~\cite{Wi} in 1932 in order to study quantum corrections to
classical statistical mechanics. Nowadays it lies at the core of the
phase-space formulation of quantum mechanics (Weyl correspondence),
and has a variety of applications in statistical mechanics, quantum
optics, and signal analysis, to name a few. In this note we consider
the localization problem of the $n$-particle Wigner distribution in
the $2n$-dimensional phase space.  We state our results precisely in
Theorem 1 below.

Equip the classical phase space ${\mathbb R}^{2n}$ with coordinates
$(x,y)$ with $x,y \in {\mathbb R}^{n}$. The Wigner quasi-probability
distribution on ${\mathbb R}^{2n}$, associated with a wave function
$\psi \in L^2({\mathbb R}^n)$ and its complex conjugate $\psi^*$, is
defined by
\begin{equation} \label {def:wigner} {\cal W}_{\psi}(x,y) =
(2 \pi)^{-n} \int_{{\mathbb R}^n}
\psi(x + {{\tau}/ 2}) \psi^*(x - \tau/2) e^{-i \tau \cdot y} \, d \tau
\end{equation}

The function ${\cal W}_{\psi}$ possesses many of the properties of a
phase space probability distribution (see e.g., \cite{J}); in
particular, it is real.
However, ${\cal W}_{\psi}$ is not a genuine probability distribution as
it can assume negative values.

The localization problem, i.e., estimating the integral of the Wigner
distribution over a subregion of the phase space, and the closely
related problem of the optimal simultaneous concentration of $\psi$
and its Fourier transform $\widehat \psi$, have received much
attention in the literature both in quantum mechanics, mathematical time-frequency analysis, and signal
processing (see e.g.~\cite{BDW1,F,G,J,J1,Le,OtM,Pool,Sl,SP,RT}, and the
references within).  Bounds on the $L^p$ norms were found in
\cite{Lieb}.  More precisely, the problem of interest for us is:

\noindent {\bf The Wigner Distribution Localization Problem:}
{\it given a measurable set $D \subset
{\mathbb R}^{2n}$, find the best possible bounds to the localization function
\begin{equation} \label{eq-def-of-energy}
 {\mathcal E}(D) := \sup_{\psi}  \int_{D}
{\cal W}_{\psi} \, dxdy,
\end{equation}
where the supremum is taken over all the functions $\psi \in
L^2({\mathbb R}^n)$ with $\|\psi\|_2 =1$. }

The quantity ${\mathcal E}(D)$ is invariant under translations in the
phase space, and under the action of the group of linear symplectic
transformations (see e.g.~\cite{W}). There is no upper bound on
${\mathcal E}(D)$; it can be infinite.  Indeed, there is a $\psi \in
L^2({\mathbb R})$ such that $\int |{\cal W}_\psi| dxdy = \infty$
\cite[sect. 4.6]{J}. An example is $\psi(x) = 1 $ if $-\frac{1}{2} <x
< \frac{1}{2}$ and $\psi(x) = 0 $ otherwise.  On the other hand, the
$L^p$ norm of ${\cal W}_\psi$ is bounded \cite{Lieb} for $p\geq 2$ and
we can use this information to show that ${\mathcal E}(D)$ is bounded
by powers of the volume $|D|$. E.g., the $L^{\infty} $ norm is at most
$\pi^{-n}$, so  ${\mathcal E}(D) \leq \pi^{-n} |D|$.

For certain $D$, however, ${\mathcal E}(D)$ is not only finite, it is
even less than 1. In~\cite{F}, Flandrin conjectured this to be true
for all convex domains, and he showed that for all centered
two-dimensional discs $B^2(r)$ of radius $r$, the standard normalized
Gaussian function $\pi^{-1/4}\exp(-x^2/2)$ is the unique maximizer
of~$(\ref{eq-def-of-energy})$.
In particular ${\mathcal E}(B^2(r)) = 1-e^{-r^2}$
(see~\cite{F}, cf.~\cite{J}).  It follows immediately from the definition of the
Wigner distribution that Flandrin's proof can be easily generalized to
higher dimensional poly-discs because the maximization problem then
has a simple product structure.  A less obvious case is the
$2n$-dimensional Euclidean ball $B^{2n}(r)$. 
The following is the generalization of Flandrin's result, and our main result:

\noindent {\bf {\large Theorem 1.}} \label{Flandrin's-result-in-higher-dimension}
{\it The standard normalized Gaussian $\pi^{-n/4}\exp(-x^2/2)$ in $L_2({\mathbb R}^n)$ is
the unique maximizer of the Wigner distribution localization problem for
any $2n$-dimensional Euclidean ball centered at the origin. In particular,
\begin{equation} \label{basic}
{\mathcal E} \left(B^{2n}(r) \right)= {\frac 1
    {\pi^n}} \int_{B^{2n}(r)}e^{-\sum_{i=1}^n(x_i^2+y_i^2)} \,
  dxdy=1-{\frac {\Gamma(n,r^2)} {(n-1)!}},
\end{equation}
where $\Gamma(s,x) = \int_x^{\infty} t^{s-1} \, e^{-t} dt$ is the upper
incomplete gamma function.}

{\it Remarks:}  (1.)  Owing to the translation
covariance of the Wigner distribution, equation~\eqref{basic}
also applies  to
a ball of radius $r$ centered anywhere in ${\mathbb R}^{2n}$. It
is only necessary to multiply the Gaussian by an appropriate linear
form $\exp(a\cdot x)$. Moreover, since the localization function~$(\ref{eq-def-of-energy})$ is invariant under 
the action of the group of linear symplectic transformations, Theorem 1 can also
be adapted to any image of the Euclidean ball under linear symplectic maps.

\quad (2.)  Another generalization is to replace the integral over the
ball with the integral over ${\mathbb R}^{2n}$, but with a weight that
is a symmetric decreasing function (i.e., a radial and non-increasing
function of  the radius $\sqrt{x^2+y^2}$). By the ``layer cake representation''
\cite[sect.~1.13]{LiebLoss} the standard Gaussian again maximizes
uniquely.

\section{Proof of
Theorem~\ref{Flandrin's-result-in-higher-dimension}}

We start with the following preliminaries. Recall that the mixed
Wigner distribution of two states $\psi_1,\psi_2 \in L^2({\mathbb
  R}^n)$ is defined by
\begin{equation}
{\cal W}_{\psi_1,\psi_2}(x,y) = (2 \pi)^{-n} \int_{{\mathbb R}^n}
\psi_1(x + \tau/2) \psi_2^*(x - \tau/2) e^{-i \tau y} \, d \tau \ .
\end{equation}
Note that in contrast to~$(\ref{def:wigner})$, ${\cal
  W}_{\psi_1,\psi_2}^{\phantom *}$ is not generally real, but,
nevertheless, Hermitian i.e., ${\cal W}_{\psi_1,\psi_2} = {\cal
  W}_{\psi_2,\psi_1}^*$.  Moreover, it is not hard to check that the
mixed Wigner distribution is sesquilinear.

Next, let $\mu=(\mu_1,\ldots,\mu_n)$ be a multiindex of non-negative
integers, and let $x \in {\mathbb R}^n$.  The Hermite functions
$H_{\mu}(x)$ on ${\mathbb R}^{n}$ are defined \cite{Th,W} to be the
product of the normalized one-dimensional Hermite functions, i.e.,
$H_{\mu}(x) = \prod_{j=1}^n h_{\mu_j}(x_j)$, where
\begin{equation}
h_k(x) = \pi^{-{\frac 1 4}} \, (k!)^{- {\frac 1 2}}
\, 2^{-{\frac k 2}} \, (-1)^k \, e^{{ {x^2}/ 2}} \,
{\frac {d^k} {dx^k}} e^{-x^2} \ .
\end{equation}
It is well known that the $\{ H_{\mu} \}$ form a complete
orthonormal system for $L^2({\mathbb R}^n)$, and that
\begin{equation} \label{eq-Hermite-functions-Schrodinger-eigenstates}
{ \mathbb H } \, H_{\mu} = |\mu| \, H_{\mu},
\end{equation}
where $|\mu|= \sum_{j=1}^n \mu_j$, and ${\mathbb H}$ is the
Schr\"{o}dinger operator $\mathbb H = - {\frac 1 2} \Delta + {\frac 1
  2}|x|^2 -{\frac n 2}$.  Here $\Delta$ denotes the standard
$n$-dimensional Laplacian.  In particular, the sesquilinearity of the
Wigner distribution implies that for any $\psi \in L_2({\mathbb
  R}^n)$, one has
\begin{equation} \label{eq-decomp-of-wave-func}
{\cal W}_{\psi} =
\sum_{\mu} \sum_{\nu} \langle \psi, H_{\mu} \rangle \,
{\langle \psi, H_{\nu} \rangle}^*\,  {\cal
W}_{H_{\mu},H_{\nu}} \ .
\end{equation}
The following lemma shows that the integral of the off-diagonal
elements of~$(\ref{eq-decomp-of-wave-func})$ over any centered ball
$B^{2n}(r)$ vanishes (cf.~\cite{J1} Section 2.3).

\begin{lemma} \label{lemma-regarding-the-off-diagonal-elements}
Let $\mu,\nu$ be two multi-indices with $\mu \neq \nu$.
Then, for every $r \geqslant 0$,  one has
 \begin{equation} \int\limits_{{B^{2n}(r)}}  {\cal
W}_{H_{\mu},H_{\nu}} dxdy=0\ .
\end{equation}
\end{lemma}
\begin{proof}[\bf Proof of Lemma~\ref{lemma-regarding-the-off-diagonal-elements}]
  It is well known (see e.g.~\cite{Le}) that for the one-dimensional
  Hermite functions $\{ h_m \}$, one has:
\begin{equation} \label{Wigner-dist-of-herm-funct}
{\cal W}_{h_j,h_k}(x_1,y_1) =
\begin{cases}
\pi^{-1} \, (k!/j!)^{1/2} \, (-1)^k \,(\sqrt 2 z_1)^{j-k} \,
e^{-(|z_1|^2)} \, L_k^{j-k}(2|z_1|^2) & \text{if } j \geq k ,\\
\pi^{-1} \, (j!/k!)^{1/2} \, (-1)^j \, (\sqrt 2 \, {\overline z_1})^{k-j}
\, e^{-(|z_1|^2)} \, L_j^{k-j}(2 |z_1|^2 ) & \text{if
} k \geq j \ .
\end{cases}
\end{equation}
Here $z_1=x_1+iy_1$, and $L_n^{\alpha}$ are the Laguerre polynomials
defined by
\begin{equation} \label{Lag-pol-def}
L_j^{\alpha}(x) = {\frac {x^{-\alpha}  e^x} {j!}} {\frac {d^j} {dx^j}}
(e^{-x} x^{j+\alpha}),
\end{equation}
for $j \geq 0$ and $\alpha > -1$. Hence the lemma holds in the
$2$-dimensional case, i.e., when $n=1$, because the integral of $z^j$
or $\overline{z}^j$ over any circle centered at the origin equals zero
when $j\neq 0$.  The higher-dimensional case follows for the same
reason from~$(\ref{Wigner-dist-of-herm-funct})$, the fact that the
Wigner distribution function ${\cal W}_{H_{\mu}, H_{\nu}}(x,y)$
is the product of ${\cal W}_{h_{m_j},h_{n_j}}(x_j,y_j)$, and the
rotation invariance of the ball $B^{2n}(r)$.
\end{proof}

An immediate corollary of
Lemma~\ref{lemma-regarding-the-off-diagonal-elements},
definition~$(\ref{eq-def-of-energy})$, and
equality~$(\ref{eq-decomp-of-wave-func})$ is
\begin{corollary} \label{Cor1} In the notation above,
\begin{equation}
{\mathcal E} \left(B^{2n}(r) \right)= \sup_{\mu}
    \int_{B^{2n}(r)} {\cal W}_{H_{\mu}} dxdy,
  \end{equation} where the
  supremum is taken over all multi-indices $\mu=(\mu_1,\ldots,\mu_n)$
  of non-negative integers.
\end{corollary}
The following lemma is the main ingredient in the proof of
Theorem~\ref{Flandrin's-result-in-higher-dimension}.
\begin{lemma} \label{main-lemma}
For any integer $\lambda \geq 0$ and multi-indices
  $\mu_1$, $\mu_2$ with $\lambda= |\mu_1|=|\mu_2|$, one has
\begin{equation}
  \int_{B^{2n}(r)} {\cal W}_{H_{\mu_1}} dxdy = \int_{B^{2n}(r)}
{\cal W}_{H_{\mu_2}} dxdy,
  \ {\rm for \ every \ } r \geqslant 0 \ .
\end{equation}
\end{lemma}
Postponing the proof of Lemma~\ref{main-lemma}, we first conclude the
proof of Theorem~\ref{Flandrin's-result-in-higher-dimension}.
\begin{proof}[{\bf Proof of Theorem~\ref{Flandrin's-result-in-higher-dimension}}]
  It follows from Corollary~\ref{Cor1} and Lemma~\ref{main-lemma}
  above that
  \begin{equation} \label{ball}
{\mathcal E} \left(B^{2n}(r) \right)=
    \sup_{\lambda} \int_{B^{2n}(r)} {\cal W}_{H_{\mu_\lambda}}
    dxdy,
\end{equation}
where $\mu_{\lambda} = (\lambda,0,\ldots,0)$, and $\lambda$ is a
non-negative integer.  Moreover,
from~$(\ref{Wigner-dist-of-herm-funct})$ and the definition of the
Wigner distribution it follow that:
\begin{equation} \label{eq1-in-proof-of-main-result}
 {\cal W}_{H_{\mu_\lambda}} (x,y)= {\frac {(-1)^\lambda} {\pi^n}} \, e^{- \sum_{i=1}^n (x_i^2+y_i^2)} \, L_\lambda \bigl (2(x_1^2+y_1^2) \bigr),
\end{equation}
where $L_\lambda(z)$ are the $\alpha=0$ Laguerre
polynomials~$(\ref{Lag-pol-def})$.  Setting $z_j= x_j + iy_j$, we
conclude that
 \begin{equation}
   \int\limits_{B^{2n}(r)} {\cal W}_{H_{\mu_\lambda}} dxdy =
   \int\limits_{\sum\limits_{j=2}^n |z_j|^2  \leq r^2}
    \!\!\!\! \!\! e^{- \sum\limits_{j=2}^n |z_j|^2}
\Bigl ( \int\limits_{ |z_1|^2  \leq r^2- \sum\limits_{j=2}^n |z_j|^2 }
   {\frac {  (-1)^\lambda} {\pi^n} } \, e^{-|z_1|^2} \,
L_\lambda  (2 |z_1|^2)  \, dz_1 \Bigr )  dz_2\cdots dz_n.
\end{equation}
On the other hand, from Flandrin's result in the $1$-dimensional
case~\cite{F}, it follows that
\begin{equation}
  \int\limits_{B^{2}(\alpha)} {\cal W}_{h_{\lambda}} \,dx_1dy_1 =    \int\limits_{ |z_1|^2  \leq \alpha^2 }
  (-1)^\lambda \, e^{-|z_1|^2} \, L_\lambda  (2 |z_1|^2)  \, dz_1
\leqslant \int\limits_{ |z_1|^2  \leq \alpha^2  }
  e^{-|z_1|^2}  \, dz_1,
\end{equation}
for every radius $\alpha \geqslant 0$.  An examination of Flandrin's proof
reveals that the inequality is strict for $\lambda >0$.
Hence, for every non-negative integer $\lambda$  one has
\begin{equation} \label{eq2-in-proof-of-main-result}
  \int\limits_{B^{2n}(r)} {\cal W}_{H_{\mu_\lambda}} \, dxdy \leqslant
    \pi^{-n} \!\!\!\!\!\!\int\limits_{ \sum\limits_{j=1}^n |z_j|^2 \leq r^2}
  e^{- \sum\limits_{j=1}^n |z_j|^2} dz_1\cdots dz_n = 1-{\frac
    {\Gamma(n,r^2)} {(n-1)!}}
\end{equation}
with equality only for $\lambda =0$. The proof of
Theorem \ref{Flandrin's-result-in-higher-dimension}  now follows
from~$(\ref{eq1-in-proof-of-main-result})$
and~$(\ref{eq2-in-proof-of-main-result})$.
\end{proof}

{\bf Remark:} The integral in (\ref{ball}\!\!) is not monotone in
$\lambda$ or in $r$ (except for $\lambda= 0$), as might have been
thought. See \cite[Fig. 2]{BDW1} and \cite{F} for interesting graphs
of these integrals as a function of $r$.

For the proof of Lemma~\ref{main-lemma} we shall need the following
preliminaries.  For a non-negative integer $\lambda $ denote
\begin{equation}
  {\mathcal H}_{\lambda} = {\rm span} \{ H_{\mu} \ ; \ |\mu| = \lambda \}
\subset L^2({\mathbb R}^n)\ .
\end{equation}
It follows from~$(\ref{eq-Hermite-functions-Schrodinger-eigenstates})$
above that the space ${\mathcal H}_{\lambda}$ consists of the
eigenfunctions of the rotation invariant Schr\"{o}dinger operator
$\mathbb H= - {\frac 1 2} \Delta + {\frac 1 2}|x|^2 -{\frac n 2} $
with eigenvalue $\lambda $.  In particular, it is a
finite-dimensional, $O(n)$-invariant subspace of $L^2({\mathbb R}^n)$
with orthonormal basis $\{H_\mu \, : |\mu|=\lambda \}$. It follows that
for every ${\mathcal R} \in O(n)$, and every $\widetilde \mu$ with $| \widetilde \mu | = \lambda$, one has:
\begin{equation}\label{rotation}
H_{\widetilde \mu}({\mathcal R}x) = \sum\limits_{\nu \, : \,
|\nu|=\lambda} c_{\nu}({\widetilde \mu, {\mathcal R}}) \, H_{\nu}(x),
\end{equation}
where the coefficients $c_{\nu}({\widetilde \mu, {\mathcal R}})$
satisfy $ \sum |c_{\nu}({\widetilde \mu, {\mathcal R}})|^2=1$.

We note the following useful fact: In order to identify which
coefficients $ c_{\nu}({\widetilde \mu, {\mathcal R}}) $ are non-zero,
it is only necessary to check  the leading powers on the two sides of
(\ref{rotation}). That is, the left side of (\ref{rotation}\!\!) defines  a
polynomial of degree $\lambda$ in the indeterminates $x_1,\dots, x_n$.
The highest degree terms are the monomials $x_1^{\mu_1 } \cdots
x_n^{\mu_n}$ with $\sum_{j=1}^n\mu_j=\lambda$, but there are also
monomials of degree lower than $\lambda$. In order to show that a given $H_\nu$
appears with a non-zero coefficient in the decomposition
(\ref{rotation}\!\!), it is only necessary to show that there is a highest
degree monomial $x_1^{\nu_1} \cdots x_1^{\nu_n}$ in the decomposition.
It is {\emph not} necessary to check the lower degree monomials; they
will appear automatically because we know that the decomposition
contains only Hermite functions of degree $\lambda$ and no others.

\begin{proof}[{\bf Proof of Lemma~\ref{main-lemma}}:]
  Fix a non-negative integer $\lambda $, and $r
  \geqslant 0$.  We consider the maximum problem
\begin{equation}
 \max\limits_{\mu \, : \, |\mu| = \lambda}
\int\limits_{B^{2n}(r)} {\cal W}_{H_{ \mu}} dxdy,
\end{equation}
and denote by $\widetilde \mu$ one of its maximizers.

From the sesquilinearity property of the Wigner distribution and Lemma
\ref{lemma-regarding-the-off-diagonal-elements}, we conclude that for every ${\mathcal R} \in O(n)$ one has:
\begin{equation}
\int\limits_{B^{2n}(r)}  {\cal W}_{H_{\widetilde \mu}({\mathcal R} \,x)} \, dxdy =
  \sum_{\nu}  | c_{\nu}({\widetilde \mu, {\mathcal R}}) |^2
  \int\limits_{B^{2n}(r)}  {\cal W}_{H_{ \nu}} \,dxdy  \ .
\end{equation}
Since $H_{\widetilde \mu}$ is a maximizer, this implies that for any
$\nu_0$ with $c_{\nu_0}({\widetilde \mu, {\mathcal R}}) \neq 0$ one
has
\begin{equation} \label{main-lemma-special-equal-maximizer}
  \int\limits_{B^{2n}(r)} {\cal W}_{H_{\widetilde \mu}} \,dxdy =
  \int\limits_{B^{2n}(r)} {\cal W}_{H_{\widetilde \mu}({\mathcal R}
    \,x)} \, dxdy = \int\limits_{B^{2n}(r)} {\cal W}_{H_{ \nu_0}}
  \,dxdy\ ,
\end{equation}
i.e., $H_{\nu_0}$ is also a maximizer.  The lemma will be proved if we
can show that, starting from any $\widetilde \mu $, we can, by a
succession of rotations and intermediate indices, finally reach
any given $\nu$.

The proof will proceed in two steps. The first is to go from
$\widetilde \mu $, by a succession of two-dimensional rotations, to
$(\lambda,0,0,\ldots,0)$ with $\lambda = \sum_{j=1}^n
\widetilde\mu_j$.

First, we show that there is a
rotation ${\mathcal R}' \in O(n)$ with
\begin{equation} \label{step-a-in-the-lemma}
  c_{\widetilde\mu'}(\widetilde \mu, {\mathcal R}') \neq 0, \ {\rm
    where} \ \ \widetilde\mu' := ((\widetilde\mu_1 +\widetilde\mu_2), 0,
\widetilde\mu_3,\ldots,
\widetilde\mu_n)\ .
\end{equation}
Thus, $\widetilde\mu' $ is also a maximizer. In a similar fashion, we
can go from $\widetilde\mu' $ to $\widetilde\mu'' $, where
$\widetilde\mu'':= ((\widetilde\mu_1
+\widetilde\mu_2+\widetilde\mu_3), 0, 0, \widetilde\mu_4,\ldots,
\widetilde\mu_n)$. Proceeding inductively, we finally arrive at the
conclusion that $(\lambda,0, \dots, 0)$ is a maximizer.

A rotation ${\mathcal R}' $ that accomplishes the first step to
$\widetilde\mu'$ is simply ${\mathcal R}' : x_1 \to (x_1+x_2)/\sqrt2,
\ x_2 \to (x_1-x_2)/\sqrt2, \ x_j\to x_j$ for $j>2$.  The monomial
$x_1^{\widetilde\mu_1} x_2^{\widetilde\mu_2}$ becomes $\frac{1}{2}
(x_1+x_2)^{\widetilde\mu_1} (x_1-x_2)^{\widetilde\mu_2} $ and this
obviously contains the monomial $x_1^{(\widetilde\mu_1
  +\widetilde\mu_2)} $ with a non-zero coefficient.

The second step is to go in the other direction, from $(\lambda, 0,
\dots, 0)$ to $(\nu_1, \nu_2, \dots, \nu_n)$ when $\sum_{j=1}^n
\nu_j = \lambda$. As before, we do this with a sequence of two-dimensional
rotations, the first of which takes us from $(\lambda, 0, \dots, 0)$
to $(\lambda -\nu_2, \nu_2, 0, \dots, 0)$.  From thence we go to
$(\lambda -\nu_2-\nu_3, \nu_2, \nu_3, 0,\dots, 0)$, and so forth.
This can be accomplished with the same
rotation as before, namely ${\mathcal R}' : x_1 \to (x_1+x_2)/\sqrt2,
\ x_2 \to (x_1-x_2)/\sqrt2, \ x_j\to x_j$ for $j>2$.

\end{proof}

\noindent {\bf Acknowledgements:} We thank P. Flandrin, A. J. E. M.
Janssen, and F. Luef for helpful discussions.  This work was partially
supported by U.S National Science Foundation grants PHY-0965859 (E. H.
L.), and (DMS-0635607) (Y. O.).

{\scriptsize

\bigskip
\noindent Elliott H. Lieb\\
Departments of Mathematics and Physics,
Princeton University. P.O. Box 708, Princeton NJ 08542, USA\\
Email: lieb@princeton.edu \bigskip

\noindent Yaron Ostrover\\
School of Mathematics, Institute for Advanced Study,
Princeton NJ 08540, USA\\
Email:  ostrover@ias.edu

}

\end{document}